\documentclass[twocolumn]{article}
\usepackage[width=17cm,height=22cm]{geometry}
\usepackage[english]{babel}
\usepackage[utf8]{inputenc}
\usepackage{fancyvrb}
\usepackage{authblk}
\usepackage{hyperref}
\usepackage{amsmath,amsthm,amssymb} 
\usepackage{bbm}
\usepackage{natbib}
\usepackage{hyperref}
\usepackage{multirow}
\usepackage{graphicx}
\usepackage{float}
\usepackage{color,soul}
\DeclareMathAlphabet{\mathpzc}{OT1}{pzc}{m}{it}

\newtheorem{theorem}{Theorem}[section]

\newtheorem{proposition}{Proposition}[section]
\newtheorem{remark}{Remark}[section]
\newtheorem{definition}{Definition}[section]

\makeatletter
\DeclareFontEncoding{LS1}{}{}
\DeclareFontSubstitution{LS1}{stix}{m}{n}
\DeclareMathAlphabet{\mathscr}{LS1}{stixscr}{m}{n}
\makeatother

\title{Omega and Sharpe ratio}
\author[1,2]{Eric Benhamou\thanks{eric.benhamou@aisquareconnect.com/dauphine.eu.\\ The authors would like to thank Francois Bertrand and Stephane Mysona for fruitful conversations about the Omega ratio.}}
\affil[1]{A.I Square Connect}
\affil[2]{Lamsade PSL}
\author[1]{Beatrice Guez}
\author[1]{Nicolas Paris}

\begin{document}
\maketitle

\begin{abstract}
Omega ratio, defined as the probability-weighted ratio of gains over losses at a given level of expected return, has been advocated as a better performance indicator compared to Sharpe and Sortino ratio as it depends on the full return distribution and hence encapsulates all information about risk and return. We compute Omega ratio for the normal distribution and show that under some distribution symmetry assumptions, the Omega ratio is oversold as it does not provide any additional information compared to Sharpe ratio. Indeed, for returns that have elliptic distributions, we prove that the optimal portfolio according to Omega ratio is the same as the optimal portfolio according to Sharpe ratio. As elliptic distributions are a weak form of symmetric distributions that generalized Gaussian distributions and encompass many fat tail distributions, this reduces tremendously the potential interest for the Omega ratio.
\end{abstract}

\medskip

\noindent\textbf{keywords}: Omega ratio, Sharpe ratio, normal distribution, elliptical distribution.

\section{Introduction}
Omega ratio has been introduced on the pledge that the Sharpe ratio and many performance measurement and risk ratios rely on two excessive simplifications. The first one states that a limited numbers of statistical characteristics can fully describe returns.  Typically mean and variance, or first and second statistical moments for Sharpe or information ratio, mean and downside standard deviation for Sortino ratio. The second one states that this performance ratios should have a return level. In an enlightening research, \cite{Keating_Shadwick_2002a} and \cite{Keating_Shadwick_2002b} introduce the Omega $\omega$ ratio, and claimed that this \textit{universal} performance measure, designed to redress the information impoverishment of traditional mean-variance analysis would address these concern. They emphasize that the Omega metric has the great advantage over traditional measures to encapsulate all information about risk and return as it depends on the full return distribution, as well as to avoid looking at a specific level as this measure is provided for all level, hence entitling each investor to look at his/her risk appetite level. Strictly speaking, the omega ratio is defined as the probability-weighted ratio of gains over losses at a given level of expected return. In financial words, this ratio determines \textit{the quality of the investment bet relative to the return threshold}. As nice as it may sound, we argue that this is oversold, as for a large class of returns distributions, that is distributions that are elliptic, we prove in this paper that the optimal portfolio according to the Omega ratio is also the optimal portfolio according to the Sharpe ratio. In other words, optimizing weights for a given portfolio of assets in order to get the optimal Omega ratio is equivalent to optimizing the portfolio weights to find the optimal Sharpe ratio. As elliptic distributions are a weak form of symmetric distributions that generalized Gaussian distributions and encompass many fat tail distributions, this kills the potential interest of the Omega ratio.

\section{Related Work}
Omega ratio have been studied in many papers producing a vast literature on it. We will review here the main papers. 
\cite{Winton_2003} is an empirical research work that looked at the historical performance of CTAs hedge funds and that put light on Omega ratio. \cite{Passow_2004} looked at the analytical property and tractabilty of the Omega ratio for Johnson distributions.  \cite{Kaffel_2010} investigated performance measurement for financial structured products, thanks to the so called Sharpe-Omega ratio that is an extension of the Omega ratio. Their originality was to compute downside risk measure using put option volatility instead of historical volatility. This allows them to take account of the asymmetry of the return probability distribution. They determined that the optimal combination of risk free, stock and call/put instruments with respect to this performance measure, is not necessarily increasing and concave as opposed to traditional optimal Sharpe ratio portfolio for the same instruments.
Similarly, \cite{Gilli_2011} studied portfolios using the Omega function, looking at their empirical performance, especially the effects of allowing short positions. Their originality was to consider short position which is traditionally ignored. They found that overall, short positions can improve risk-return characteristics of a portfolio but mitigated this findings with the constraints involved in short positions that often carries additional constraints in terms of transactions costs and liquidity.

\cite{Bertrand_2011} analyzed the performance of two main portfolio insurance methods, the OBPI and CPPI strategies, using downside risk measures, thanks to Omega measure. They showed that CPPI strategies perform better than OBPI. \cite{Kapsos_2014} looked at the maximum Omega ratio. They established that it can be computed as a linear program optimization problem. While the Omega ratio is theoretically a nonconvex function, \cite{Kapsos_2014} l showed that this can be reformulated as a convex optimization problem that can be solved thanks to a linear program. This convex reformulation for Omega ratio maximization can be seen as a direct analogy of the mean-variance framework for the Sharpe ratio maximization and paved the way for our work that will show the strong connection between Omega and Sharpe ratio. \cite{vanDyk_2014} provided a nice empirical research on the difference of ranking between Sharpe and Omega ratio. They compared the ranking of 184 international long/short (equity) hedge funds, over the period January 2000 to December 2011 using their monthly returns. They concluded that Omega ratio does indeed provide useful additional information to investors compared to the one only provided by Sharpe ratio alone.
\cite{Belles-Sampera_2014} generalized Omega ratio in a so called GlueVaR risk measure. It combines Value-at-Risk and Tail Value-at-Risk at different tolerance levels and has analytical closed-form expressions for commonly used distribution like Normal, Log-normal, Student-t and Generalized Pareto distributions. They showed that under certain condition, a subset of GlueVaR risk measures fulfils the property of tail-subadditivity
\cite{Sharma_2016} worked on the threshold to be used in portfolio optimization with Omega ratio. In order to maximize the Omega ratio for the overall portfolio, one needs to consider a threshold point to compute the Omega ratio as optimizing at all thresholds is not realistic. They decided to use the conditional value-at-risk at an $\alpha $ confidence level ${\mathrm{CVaR}_{\alpha }}$ of the benchmark market.  They argue that this $\alpha $-value reflects the attitude of an investor towards losses. Like in \cite{Kapsos_2014}, this formulation can be cast as a 
linear program for mixed and box uncertainty sets and a second order cone program under ellipsoidal sets, and hence becomes tractable. They showed that the optimal portfolios resulting from the Omega-${\mathrm{CVaR}_{\alpha }}$ model exhibit a superior performance compared to the classical ${\mathrm{CVaR}_{\alpha }}$ model in the sense of higher expected returns, Sharpe ratios, modified Sharpe ratios, and lesser losses in terms of ${\mathrm{VaR}_{\alpha }}$ and ${\mathrm{CVaR}_{\alpha }}$ values. 
\cite{Guo_2016} worked on the property for one asset to have a higher Omega ratio than a second one. They showed that second-order stochastic dominance (SSD) and/or second-order risk seeking stochastic dominance (SRSD) alone for any two prospects is not sufficient to imply that the Omega ratio of one asset is always greater than that of the other one. Indeed, they proved that the second-order stochastic dominance only implies higher Omega ratios only for thresholds that are between the mean of the smaller-return asset and the mean of the higher-return asset. When considering first-order stochastic dominance, the restriction on the thresholds does not apply and first-order stochastic dominance implies preference of the corresponding Omega ratios for any threshold.
\cite{Krezolek_2017} introduced an extension of Omega ratio called GlueVaR risk measure and illustrated this on metals market investments. GlueVaR risk measures combine Value-at-Risk and Tail Value-at-Risk at different tolerance levels and have analytical closed-form expressions for the most frequently used distribution functions in many applications, i.e. Normal, Log-normal, Student-t and Generalized Pareto distributions. 
\cite{Metel_2017} is an illuminating paper as it is the first to notice the correspondence between Sharpe and Omega ratio under jointly elliptic distributions of returns. Compared to our work, their proof is more convoluted and does not emphasize the important fact that elliptic distributions satisfy some symmetry properties that validates the proof.
\cite{Rambo_2017} ranked fund returns and compared results obtained with those obtained from the Sharpe ratio over two periods: 2001 to 2007 and 2008 to 2013. They found that Omega ratio provides far superior rankings.
\cite{Guo_2018} investigated whether there is any Sharpe or Omega ratio rule that prove that one asset outperforms another one. They found that Sharpe ratio rule is not able to detect preference under general strong dominance cases. In contrast to mean-variance rule implied by Sharp ratio that does not work under first order dominance, Omega ratio can help detecting better performance under first order stochastic dominance.
\cite{Caporin_2018} is a nice work on the overall developments on Omega ratio over the last two decades. They emphasized two flaws of Omega ratio. First, Omega ratio does not comply with Second-order Stochastic Dominance as already noted by \cite{Guo_2016}. Second, trade-off between return and risk corresponding to the Omega measure, is highly influenced by the mean return. They illustrated their work on long-only asset and hedge fund databases to confirm the issues with Omega ratio.
\cite{Bernard_2019} proved that in a continuous-time setting, the problem of maximizing Omega ratio is ill-posed and leads to excessive risk taking. They investigated if additional constraints could offset the Omega ratio risk problem but concluded that this was not obvious and caution should be taken when using Omega ratio for making asset allocation decisions.

\section{Contribution and paper outline} \label{sec:Outline}
In this paper, we first present elliptical distributions and some of their key properties in section \ref{sec:ED}. We emphasize that these distributions encompasses many standard distributions and are a generalization of normal distribution in terms of symmetry assumptions. \ref{sec:Omega} we provide the exact computation of the Omega ratio for the normal distribution and show that optimizing the Omega ratio for a portfolio is indeed similar to optimizing the Sharpe ratio. We extend this result by noticing that the equivalence between Omega and Sharpe ratio relies on properties of symmetry of the normal distribution and can be extended to elliptical distributions. We show that this result can be proved very rapidly thanks to the specific nature of the Omega ratio. We conclude that Omega ratio is oversold and is in most cases not providing additional information compared to stochastic dominance.

\section{Elliptical distributions} \label{sec:ED}
It is well known that a normal distribution is fully characterized by its first and second moments. It is less well known or at least less emphasized that the normal distribution also has a very nice property in terms of symmetry with respect to its first and second moments. If one plots iso-density curves in two or three dimension for the multi variate normal, one would obtain ellipsoid as illustrated by \ref{fig:contour1} and \ref{fig:contour2}. This symmetry according to axes leads to the so called elliptic distributions called like this because they are elliptically contoured distributions. Elliptic distribution were introduced by \cite{Kelker_1970} and further studied in \cite{Cambanis_1981} and by \cite{Fang_1990}. 

\begin{figure}[!ht]
\centering
\includegraphics[width=7cm]{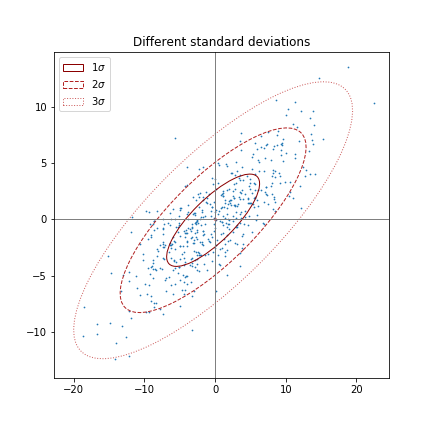}
\vspace{-0.5cm}
\caption{Contour plot for normal with positive correlation} \label{fig:contour1}
\end{figure}

\begin{figure}[!ht]
\centering
\includegraphics[width=7cm]{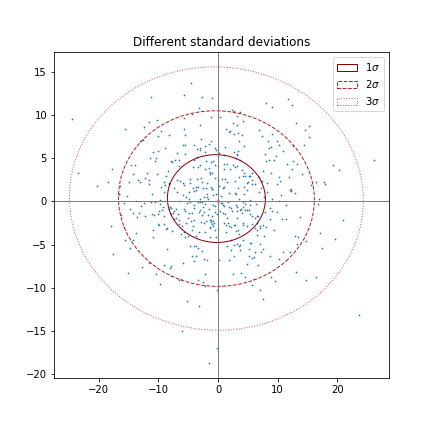}
\vspace{-0.5cm}
\caption{Contour plot for normal with weak correlation} \label{fig:contour2}
\end{figure}

\subsection{spherical distributions}
A first intuitive idea is to define spherical distributions. Let us denote by $\mathbf{U} \in \mathbb{R}^{n \times n}$ an orthogonal matrix, that is $\mathbf{U} \mathbf{U}^T = \mathbf{U}^T \mathbf{U} =  \mathbf{I}_n$. The matrix $\mathbf{U}$ defines an orthogonal linear transformation $L_{\mathbf{U}}: \mathbb{R}^{n \times n} \mapsto   \mathbb{R}^{n \times n}$ defined by $\mathbf{X} \rightarrow \mathbf{U} \mathbf{X}$. Please note that this is not necessarily a rotation as the matrix associated with the transformation is not necessarily a rotation matrix. It does not necessarily comply with $\det(\mathbf{U}) = 1$.

\begin{definition}
A random vector $\mathbf{X} = (X_1,\ldots, X_n)$ has a \textit{spherical distribution} if it is invariant according to any orthogonal linear transformation, that is for any  $\mathbf{U} \in \mathbb{R}^{n \times n}$ such that  $\mathbf{U} \mathbf{U}^T = \mathbf{U}^T \mathbf{U} =  \mathbf{I}_n$, we have
$\mathbf{U} \mathbf{X} \sim \mathbf{X}$
where $\mathbf{X} \sim \mathbf{Y}$ means that $\mathbf{X}$ and $\mathbf{Y}$ are equal in law or in distribution. 
\end{definition}

\begin{remark}
In particular, the distribution of a spherical distributed random variable, $\mathbf{X}$, is invariant under rotations as rotations matrices are special cases of orthogonal matrices.

\end{remark}
It can be shown easily (see \cite{Fang_1990})  that $\mathbf{X}$ is a spherical distribution is equivalent to the existence of a function $\phi(.)$ such that for any $t \in \mathbb{R}^{n}$, we have 
\begin{equation}
\Phi_{\mathbf{X}}(t) = \phi(t^T t) = \phi(t_1^2 + \ldots + t_n^2)
\end{equation}
where $\Phi_X(t)$ denotes the characteristic function $\Phi_{\mathbf{X}}(t) = \mathbb{E}[ e^{i t^T \mathbf{X} } ]$. The function $\phi$ is referred to as the generator of the distribution and it is quite common to write $\mathbf{X} \sim \mathcal{S}_n(\phi)$ where $\mathcal{S}_n$ denotes the spherical distribution in $ \mathbb{R}^{n}$. A nice theorem that translates into an enlightening geometrical interpretation is that there exists a representation form for any spherical distribution. This is the following theorem known as the spherical representation theorem: 

\begin{theorem}
A random vector $\mathbf{X} = (X_1,\ldots, X_n)$ has a spherical distribution if and only if there exists $\mathbf{S}$ 
a random variable uniformly distributed on the unit sphere $\mathcal{S}^{n-1} = \left\{ s \in \mathbb{R}^{n} : s s^T = 1 \right\}$ 
and $\mathcal{R} \ge 0$ a random variable independent of $S$ such that 
\begin{equation} 
\mathbf{X} \sim \mathcal{R}  \mathbf{S}.
\end{equation}
If in addition, the scalar random variable $\mathcal{R} $ has a finite second moment $\mathbb{E}[\mathcal{R} ^2] < \infty$, then the first two moments of $\mathbf{X}$ exist and are given by:
\begin{equation}
\mathbb{E}  [\mathbf{X} ] = 0, \qquad \mathbb{C}ov(\mathbf{X} )= 1/ n \,\, \mathbb{E}[\mathcal{R} ^2]  \mathbf{I}_n. 
\end{equation}
\end{theorem}

\begin{remark}
The random variable $\mathcal{R}$ is often referred to as the  \textit{radius} of the spherical distribution and can be interpreted geometrically in dimension two or three as the mean radius of the ellipsoid, also seen as the radius the major axis of the ellipse. 
\end{remark}

\begin{remark} 
Spherical distributions means that there are distributions that comply with some symmetrical properties (invariance along any rotation but also along any orthogonal linear transformation). This should not be confused with distributions on the sphere like the Fisher–Bingham or Kent distribution, the Von Mises–Fisher distribution or even the Bingham distribution that are sometimes incorrectly called spherical distributions! They are not at all spherical distributions but rather distributions on the sphere which means that they are probability distribution such that the probability assigned to the unit sphere is 1 and 0 elsewhere. In general, a spherically distributed random vector does not have its density support restricted to the sphere. It does not, either necessarily possess a density. However, if it does, the marginal densities of dimension smaller than $n-1$ are continuous and the marginal densities of dimension smaller than $n-2$ are differentiable (except possibly at the origin in both cases). Uni-variate marginal densities for $n$ greater than $2$ are non-decreasing on $(-\infty, 0)$ and non-increasing on $(0, \infty)$.
\end{remark}

\begin{remark} 
A typical example of a spherical distribution is the multi variate normal distribution with covariance matrix proportional to the identity. 
The multivariate t-distribution is a typical example of a fat tailed spherical distribution. Let $Z$ be a multi variate spherical normal distribution $Z \sim \mathcal{N}_n(0,\mathcal{I}_{n})$ and $\mathcal{R}$ be a chi squared distribution  $\mathcal{R} \sim \chi^2_k$ with  $k$ degrees of freedom independent of $Z$ .The random vector \begin{equation}
Y=\sqrt{k} \ \frac{Z}{\mathcal{R}}
\end{equation} has a multivariate $t$-distribution with $k$ degrees of freedom. Thanks to this writing, it is easy to see that its belongs to the family of n-dimensioned spherical distributions.
\end{remark}

Obviously, for multi variate Gaussian, imposing that the covariance matrix is proportional to the identity matrix is too restrictive. This naturally leads to the extension of spherical distributions that are elliptical distributions.

\subsection{Elliptical distributions}

\begin{definition}
A $(n \times 1)$ random vector $\mathbf{X}$ is said to have an elliptical distribution with parameters $\mu$ a $(n \times 1)$ constant vector and $\Sigma$ a $(n \times n)$ constant matrix if $\mathbf{X}$ has the same distribution as 
$\mu + \Lambda^\top \mathbf{Y}$, where $\mathbf{Y} \sim \mathcal{S}_n(\phi)$ and $\Lambda$ is a $(k \times n)$ matrix such that 
$\Lambda^\top \Lambda=\Sigma$ with $\mathop{\rm {rank}}(\Sigma)=k$. We shall write $\mathbf{X} \sim EC_n(\mu,\Sigma,\phi)$.
\end{definition}

\begin{remark} Elliptical distributions can be seen as an extension of the multi variate normal distribution denoted by $\mathcal{N}_p(\mu,\Sigma)$. 

Let $\mathbf{X}$ be a  multinormal distribution $\mathbf{X} \sim \mathcal{N}_n (\mu, \Sigma) $. Then $\mathbf{X} \sim \mathcal{E}_n(\mu,\Sigma,\phi)$ and $\phi(u)=\exp{(-u/2)}$. As the density surface for the multivariate normal distribution is given by $f(x) = \mathop{\rm {det}}(2 \pi \Sigma)^{ -\frac{1}{2} }
\exp\{-\frac{1}{2} (x-\mu)^\top \Sigma^{-1} (x-\mu) \}$, it is easy to see that this density is constant on ellipses (see for instance \ref{fig:contour1}). This explains why these distributions have been called \textit{elliptical}.
\end{remark}

Below are  summarized the main properties of elliptical distributions:

\begin{theorem}
Elliptical random vectors $\mathbf{X}$ have the following properties:
\begin{itemize}
\item Basic properties: any linear combination of elliptically distributed variables is elliptical and any of its marginal distributions is also elliptical.

\item Representation form: a scalar function $\phi(.)$ can determine an elliptical distribution $\mathcal{E}_n(\mu,\Sigma,\phi)$ for every $ \mu \in \mathbb{R}^p$ and $\Sigma \geq 0$ with $\mathop{\rm {rank}}(\Sigma)=k$ iff $\phi(t^\top t)$ is a $p$-dimensional characteristic function. If there exist two representation forms of $\mathbf{X}$ such that $\mathbf{X} \sim \mathcal{E}_n(\mu,\Sigma,\phi)$ and $\mathbf{X} \sim \mathcal{E}_n(\mu^*,\Sigma^*,\phi^*)$, then there exists a constant $c > 0$ such that
\begin{displaymath}
\mu=\mu^*, \quad \Sigma=c\Sigma^*, \quad
\phi^*(.) = \phi(c^{-1}.).
\end{displaymath}

\item Characteristic function and representation: the characteristic function of $X$, $\Phi_{\mathbf{X}}(t)$ is of the form
\begin{displaymath}\label{sec:characteristic_function}
\Phi_{\mathbf{X}}(t) = e^{\textrm{\bf i}t^\top \mu}\phi(t^\top \Sigma t)
\end{displaymath}
for a scalar function $\phi$.
$\mathbf{X} \sim \mathcal{E}_n(\mu,\Sigma,\phi)$ with $\mathop{\rm {rank}}(\Sigma)=k$ iff $\mathbf{X}$ has the same distribution as:
\begin{gather}\label{eq:elliptic_representation}
\mu + \mathcal{R} \Lambda^\top U^{(k)}
\end{gather}
where $\mathcal{R} \geq 0$ is a random scalar variable independent of $U^{(k)}$. $U^{(k)}$ is a random vector distributed uniformly on the unit sphere surface in $\mathbb{R}^k$ and $\Lambda$ is a $(k \times n)$ matrix such that $\Lambda^\top \Lambda=\Sigma$.

\item Assume that $\mathbf{X} \sim \mathcal{E}_n(\mu,\Sigma,\phi)$ and $E (\mathcal{R} ^2) < \infty$. Then its first two moments exist and are given by:
\begin{displaymath}\mathbb{E} (X) = \mu, \quad \mathbb{C}ov(X) =
\frac{\mathbb{E} (\mathcal{R}^2)}{\mathop{\rm {rank}}(\Sigma)}\Sigma= - 2\phi^\top (0)\Sigma.\end{displaymath}

\item Assume that $\mathbf{X} \sim \mathcal{E}_n(\mu,\Sigma,\phi)$ with $\mathop{\rm {rank}}(\Sigma)=k$. Then 
\begin{displaymath}Q(X)= (X - \mu)^\top \Sigma^-(X - \mu)\end{displaymath}
has the same distribution as $\mathcal{R}^2$ in equation \ref{eq:elliptic_representation}

\item An elliptic distribution has a density function of the form
\begin{displaymath}f(x)=C\cdot g\left((x-\mu )^\top \Sigma ^{-1}(x-\mu )\right) \end{displaymath}
where $C$ is the normalizing constant, $x$ is an $n$-dimensional random vector with median vector $mu$  (which is also the mean vector if the latter exists), and $\Sigma$ is a positive definite matrix which is proportional to the covariance matrix if the latter exists.
\end{itemize}
\end{theorem}

\begin{remark}
The second item of this theorem states that $\Sigma, \phi, \Lambda$ are not unique, unless we impose a condition on the determinant, that is $\mathop{\rm {det}}(\Sigma)=1$.
\end{remark}

\begin{remark}
The distribution function of $\mathcal{R}$  is a key characteristic of the elliptical distribution and can lead to elliptical distributions that share most of their characteristic in common. 
Hence it is called the generating variate of the elliptical distribution family. 
Using extreme value theory, it is useful to characterize fat tailed elliptical distributions. 
If the generating variate $\mathcal{R}$ belongs to the maximum domain of attraction of the Frechet distribution (\cite{Embrechts_2003}), this writes as $\overline{F}_{\mathcal{R}}= \lambda(x) \times x^{-\alpha}$ 
for all $x > 0$, where $\alpha > 0$ and $\lambda$ is a slowly varying function (\cite{resnick2014extreme}). 
The parameter $\alpha$ is then called the tail index of the generating distribution function $F_{\mathcal{R}}$ which corresponds also to the tail index of the regularly varying random vector $\mathbf{X}$ (\cite{Hult_2002}). 
This shows in particular that multivariate elliptical distributions allow for heavy tails while encompassing at the same time a simple linear dependence structure from the multi variate normal. Hence, in addition to the multi variate normal distribution, the multi variate $t$-distribution 
\cite{Fang_1990}, the symmetric generalized hyperbolic distribution \cite{BarndorffNielsen_1982}, the sub-Gaussian $\alpha$-stable distribution (Rachev and Mittnik, 2000, p. 437) are elliptical distributions.
\end{remark}

\begin{remark}
Because of their representation, elliptical distributions generalizes nicely Gaussian properties. This is because their characteristic function writes as 
\begin{equation}
\Phi_{\mathbf{X}}(t) = e^{\textrm{\bf i}t^\top \mu}\phi(t^\top \Sigma t)
\end{equation}
which is a weakened form of the multi variate Gaussian distribution that is given by
\begin{equation}
\Phi_{\mathbf{X}}(t) = e^{\textrm{\bf i}t^\top \mu} \exp(-1/2 t^\top \Sigma t)
\end{equation}
\end{remark}

\begin{remark}
Because of their nice property of generalizing multi variate Gaussian distribution with potentially fat tailed distributions, elliptical distributions are meaningful for financial data modeling. 
The theory of portfolio optimization developed by \cite{Markowitz_1952} and continued by \cite{Tobin_1958}, \cite{Sharpe_1963} is the basis of modern portfolio risk management. It relies on the Gaussian distribution hypothesis and its quintessence is that the portfolio diversification effect depends essentially on the covariance matrix, i.e. the linear dependence structure of the portfolio components. Elliptical distributions generalize nicely this portfolio theory as they cope well with linear transformations. In particular, if the returns on all assets available for portfolio formation are jointly elliptically distributed, then all portfolios can be characterized completely by their location and scale – that is, any two portfolios with identical location and scale of portfolio return have identical distributions of portfolio return. Various features of portfolio analysis, including mutual fund separation theorems and the Capital Asset Pricing Model, can be easily extended for all elliptical distributions, making this kind of distributions appealing.
\end{remark}

\section{Omega ratio} \label{sec:Omega}
\begin{definition}
As presented in \cite{Keating_Shadwick_2002a} and \cite{Keating_Shadwick_2002b}, for an asset whose return $r$ has a cumulative probability distribution function $F$ and $\theta$ is the target return threshold defining what is considered a gain versus a loss, the Omega ratio is defined as
\begin{equation}
\Omega (\theta )= \frac {\int _{\theta }^{\infty }[1-F(r)]\,dr} {\int _{-\infty }^{\theta }F(r)\,dr}
\end{equation}
\end{definition}

When $\theta$ is set to zero the gain-loss-ratio by Bernardo and Ledoit arises as a special case \cite{Bernardo_2000}. The selling point of the omega ratio compared to Sharpe ratio and other traditional risk ratio is that at first sight, it seems to depend on the entire return distribution through the cumulative probability distribution function $F$ as well as not rely on any particular moments in terms of value and even existence, making it intellectually very attractive. Graphically, for a $\theta$ value of $2.70$ percent and the cumulative distribution given by figure \ref{fig:omega}, the $\Omega (\theta )$ ratio is defined as the ratio of the red area over the blue area.

\begin{figure}[!ht]
\centering
\includegraphics[width=7cm]{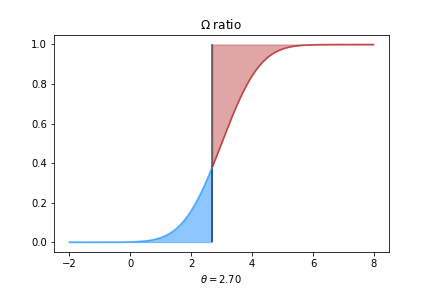}
\vspace{-0.5cm}
\caption{Graphical interpretation of Omega ratio  $\Omega (\theta )$ defined as the ratio of the shaded red over the blue area} \label{fig:omega}
\end{figure}

Omega ratio has been created to be able to compare different assets in terms of their risk profile. Figures \ref{fig:omega_portfolio} and \ref{fig:omega_portfolio2}  provide in log scale the omega ratio for four portfolio. The omega ratio is a real function that is infinite for large negative value of $\theta$ decreases with $\theta$ and tends to 0 for large positive value of $\theta$. The perfect case is to have one curve of a portfolio above all the other one for any value of $\theta$. However, this particular case of \textit{Pareto} optimality (meaning the curve is above all other curves for any value of $\theta$) is extremely rare and in practice, curves cross each other making the call to select one asset or portfolio among the other ones harder.

\begin{figure}[!ht]
\centering
\includegraphics[width=7cm]{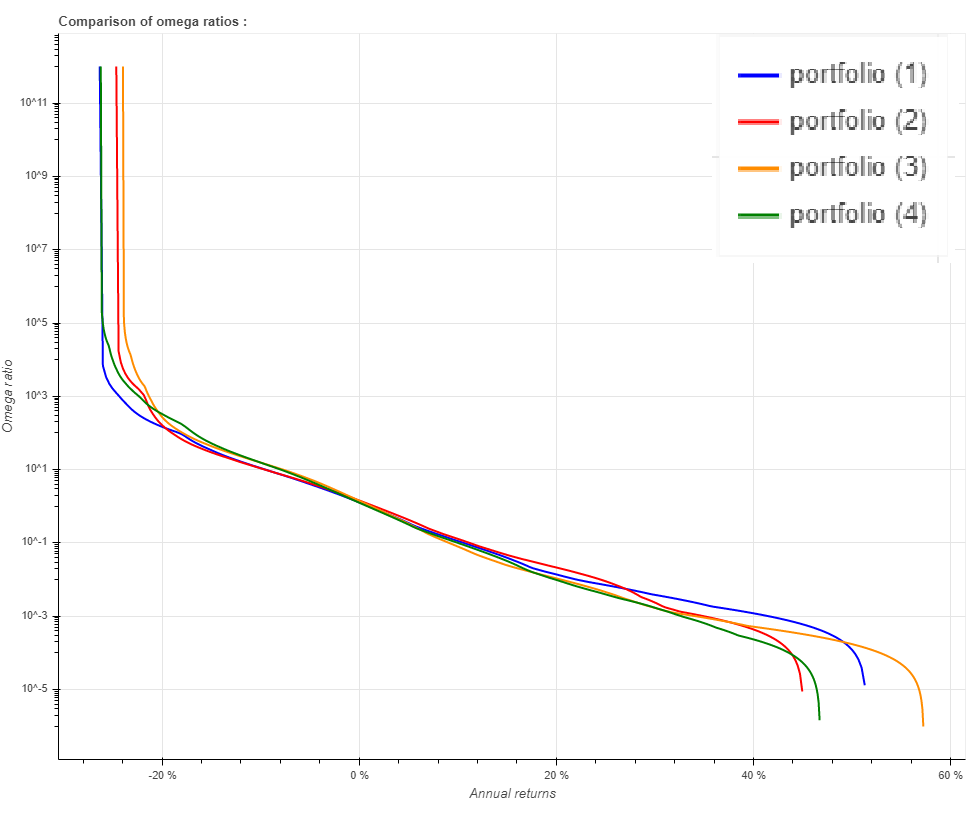}
\vspace{-0.5cm}
\caption{$\Omega (\theta )$ ratio for four portfolios in practice} \label{fig:omega_portfolio}
\end{figure}

\begin{figure}[!ht]
\centering
\includegraphics[width=7cm]{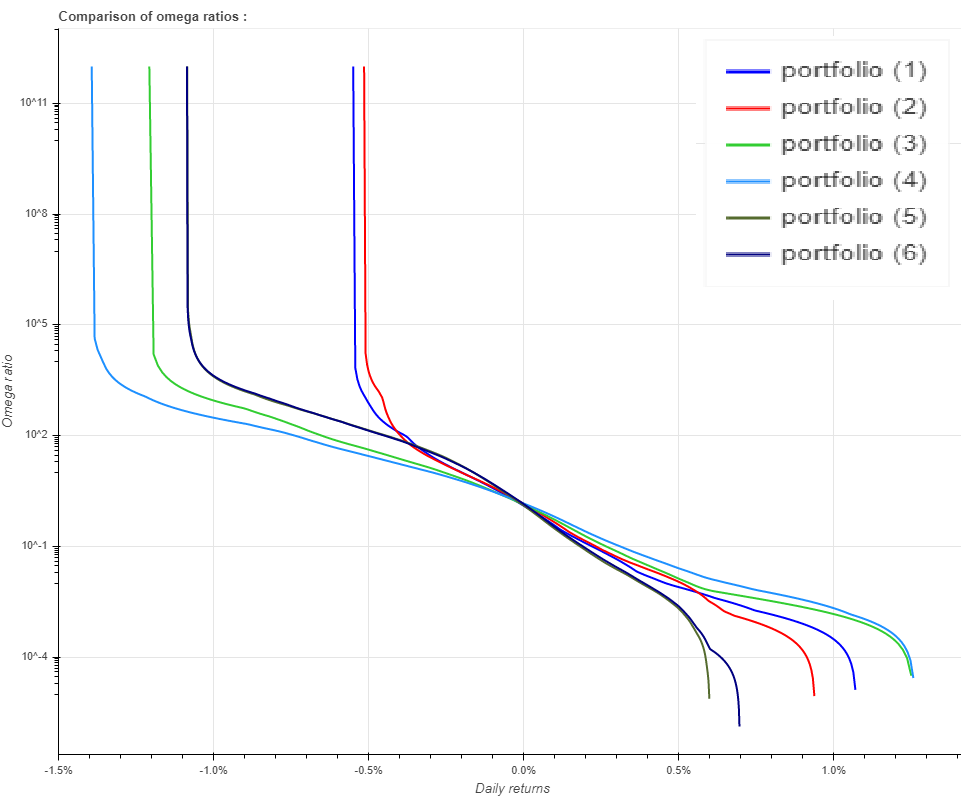}
\vspace{-0.5cm}
\caption{$\Omega (\theta )$ ratio for six portfolios in practice} \label{fig:omega_portfolio2}
\end{figure}

\begin{remark}
It is worth mentioning that the Omega ratio can be meaningless if for instance the numerator term $ \int _{\theta }^{\infty }[1-F(r)]\,dr$ or the denominator term 
$\int _{-\infty }^{\theta }F(r)\,dr$ are undefined. Surprisingly, this is ignored in the literature about the $\Omega (\theta )$ ratio. 
A typical example of a distribution that has an undefined  $\Omega (\theta )$  ratio is a power law $F \sim | r|^{\alpha}$ with $1 < \alpha \leq 2$ .
\end{remark}

\begin{remark}
The case of the power law for $\alpha = 2$ is illuminating. A corresponding elliptical distribution is the Cauchy distribution whose density is given by $f(r) = \frac {1}{\pi \sigma \, \left[1+\left({\frac {r-\mu}{\sigma}}\right)^{2}\right]}$, where $\mu$ is the location parameter and $\sigma$ the scale parameter. Its cumulative density function is given by $F(r) = \frac {1}{\pi } \arctan \left({\frac {r-\mu }{\sigma }}\right)+\frac {1}{2}$. To keep things simple, let us take the case of $\mu=0$ and $\sigma=1$, which is the normalized Cauchy distribution. The numerator of the omega ratio is then given as the limit for $A \rightarrow +\infty$ of 
\begin{eqnarray}
 \int _{\theta }^{A}[1-F(r)]\,dr  &\hspace{-0.5cm} = & \hspace{-0.5cm} \int _{\theta }^{A} \frac {1}{2} - \frac {1}{\pi } \arctan (r) dr  \\
&\hspace{-0.5cm} = & \hspace{-0.5cm}\left[\frac {r}{2} -  \frac {r\arctan (r) - \frac{\ln(1+r^2) }{2}}{\pi } \right]_{\theta }^{A}  \nonumber \\
&\hspace{-0.5cm}  \underset{A \rightarrow  +\infty}{\rightarrow}  & \hspace{-0.5cm} +\infty
\end{eqnarray}
This shows that whenever, we speak about Omega ratio, we need to impose that the two terms (numerator and denominator) are well defined. This is the subject of the following proposition
\end{remark}

\begin{proposition}\label{prop0}
The  $\Omega (\theta )$  ratio terms are defined if and only if the right and left tails of the cumulative distribution are dominated by $\frac{1}{|r|^{\alpha}}$ with $\alpha > 1$. This means in particular that the existence of the $\Omega (\theta )$  ratio  implies that $\underset{r \rightarrow  -\infty}{\lim} r F(r) = 0$.
\end{proposition}
\begin{proof}
Given in appendix \ref{app:proof0}
\end{proof}

In contrast, the Sharpe ratio is for a given level $\theta$ is defined as follows:

\begin{definition}
The Sharpe ratio (for the reference level $\theta$) $S (\theta )$ is defined as the excess return over $\theta$ divided by the standard deviation
\begin{equation}
S (\theta )= \frac {\mu-\theta}{\sigma}
\end{equation}
\end{definition}

As said in the introduction, Sharpe and Omega ratio have been strongly opposed. We shall see this is not completely correct. 
In the particular case of a normal probability distribution function $F$, it is fairly easy to compute the $\Omega (\theta )$. Let us denote by $\psi$ the standard normal probability distribution function $\psi(x)=\frac{1}{\sqrt{2\pi}} exp(-x^2/2)$ and $\Psi(x)$, the standard normal cumulative probability distribution function 
$\Psi(x) = \int _{-\infty}^{x}\psi(u)du$. We have

\begin{proposition}\label{prop1}
If returns follow a normal probability distribution function $\mathcal{N}(\mu,\sigma)$, the  $\Omega (\theta )$  ratio is given by
\begin{equation}\label{eq:prop1}
\Omega (\theta ) = 1 + \frac{1}{  \frac{\psi( S(\theta) )}{S(\theta) } -  \Psi(-S(\theta) ) }
\end{equation}
\end{proposition}
\begin{proof}
Given in appendix \ref{app:proof1}
\end{proof}

If we explicitly expand equation \eqref{eq:prop1} for the normal distribution, we get
\begin{equation}\label{eq:normal_explicit}
\Omega (\theta ) = 1 + \frac{\sqrt{2 \pi} }{ \frac{exp( -S(\theta)^2 / 2)  }{ S(\theta)  } - \int_{-\infty}^{-S(\theta)} exp(-u^2/2) \, du }
\end{equation}

\begin{proposition}\label{prop2}
If returns follow a normal probability distribution function $\mathcal{N}(\mu,\sigma)$, the $\Omega (\theta )$  ratio is an increasing function of $S(\theta) $
\end{proposition}
\begin{proof}
Given in appendix \ref{app:proof2}
\end{proof}

\textbf{It follows that maximizing the Omega ratio at level $\theta$ is equivalent to maximize the Sharpe ratio for a risk free rate $\theta$}. 
Moreover, if we are looking at returns at a given level of volatility, it means that we impose that $\sigma$ is a constant. In this particular case, as the Sharpe only depends on the first moment: $S(\theta) = \frac{\mu - \theta}{\sigma} = \frac{\mu - \theta_0}{\sigma} +  \frac{\theta_0 - \theta}{\sigma}$, maximizing the Sharpe does not depend on the corresponding risk free rate level $\theta$ and can be done at the risk free rate. We therefore obtain the important result given by the following proposition

\begin{proposition}\label{prop2bis}
Under a target volatility constraint, and for a normal distribution, maximizing the $\Omega (\theta )$ is strictly the same as maximizing the Sharpe ratio! \end{proposition}

We will see in the following that the property is not restricted to normal distribution but can easily be extended to elliptic distributions as we have just used symmetry properties of the distribution. More precisely, we can compute explicitly the  $\Omega (\theta )$  ratio as follows:

\begin{proposition}\label{prop3}
If the $\Omega (\theta )$  ratio  is well defined and if returns follow an elliptic probability distribution function proportional to $g\left(-\frac{1}{2} \left[ \frac{u-\mu}{\sigma} \right]^2 \right)$, writing $\psi(u) = C g( -\frac{1}{2} u^2)$ and $\Psi(v) = \int_{-\infty}^{v} C g( -\frac{1}{2}  u^2) \, du$, the corresponding re-normalized probability  and cumulative  probability distribution functions, as well as the following function $G(u)$:
\begin{eqnarray}\label{eq:Gdef}
G(u)= \int_{-\infty}^u C g(x)\, dx
\end{eqnarray}
then the $\Omega (\theta )$  ratio is given by
\begin{equation}\label{eq:prop3}
\Omega (\theta ) = 1 + \frac{1}{ \frac{G( -\frac{1}{2}  S(\theta)^2 )}{S(\theta) } - \Psi(-S(\theta) )}
\end{equation}
\end{proposition}
\begin{proof}
Given in appendix \ref{app:proof3}
\end{proof}

\begin{remark}
The re-normalizing constant $C$ is defined as
$\int_{-\infty}^{\infty} C g( -\frac{1}{2}  u^2) \, du = 1$, or equivalently 
\begin{equation}
C = \frac{1}{\int_{-\infty}^{\infty} g( -\frac{1}{2}  u^2) \, du }
\end{equation}
\end{remark}

\begin{remark}
In the case of the normal distribution, the function $g$  is simply the exponential function.
\end{remark}

\begin{proposition}\label{prop4}
For elliptic distribution, the $\Omega (\theta )$  ratio is an increasing function of $S(\theta)$.
\end{proposition}
\begin{proof}
Given in appendix \ref{app:proof4}
\end{proof}

We come now to the central property
\begin{proposition}\label{prop4final}
Under a target volatility constraint, and for an elliptic distribution, maximizing the $\Omega (\theta )$ is strictly the same as maximizing the Sharpe ratio! \end{proposition}

\section{conclusion}
In this paper, we have shown that under  elliptic distributions assumptions, with tail dominated by $1/r^{\alpha}$ for some $\alpha>1$, the $\Omega (\theta )$ ratio is well defined. Under these assumption, the $\Omega (\theta )$ ratio can be computed in closed form and is an increasing function of the Sharpe ratio for a risk free rate taken to $\theta$. Assuming that we impose a target volatility constraint, meaning we only care about portfolio for a given level of volatility, maximizing the  $\Omega (\theta )$ ratio is not any more dependent on the parameter $\theta$ and is obtained by maximizing the true Sharpe ratio for the real risk free rate. This means that under target volatility constraint, first order dominance (higher first moment) provides the best Sharpe and the best Omega ratio.

\bibliography{biblio}

\section{Appendix}

\subsection{Proof of Proposition \ref{prop0}}\label{app:proof0}
The case of the numerator and denominator are the same. Hence, let us only prove the case for the denominator. Using the comparison test for improper integral and the fact that $F$ is always positive, we have that $\int _{-A}^{\theta }F(r)\,dr$ is well defined for $A \rightarrow  \infty$ if and only $F$ is dominated by some function $\frac{1}{| r | ^{\alpha}}$ for $\alpha > 1$. To conclude the proof, let us first notice that
\begin{eqnarray}
\underset{r \rightarrow  -\infty}{\lim} r F(r) &=& - \underset{r \rightarrow  \infty}{\lim}  r F(-r)
\end{eqnarray}

As $F$ is dominated by  $\frac{1}{|r|^{\alpha}}$ for $\alpha > 1$ for large negative value of $r$, we have easily that 
\begin{eqnarray}
\underset{r \rightarrow  \infty}{\lim}  r F(-r) & \leq& \underset{r \rightarrow  \infty}{\lim} \frac{1}{r^{\alpha-1}} \\
&  \underset{r \rightarrow  \infty}{\rightarrow } & 0
\end{eqnarray}
which shows that $\underset{r \rightarrow  -\infty}{\lim} r F(r) = 0$
\qed.

\subsection{Proof of Proposition \ref{prop1}}\label{app:proof1}
If returns follow a normal probability distribution function $\mathcal{N}(\mu,\sigma)$, the Omega ratio, $\Omega (\theta )$, writes as:
\begin{eqnarray}
\Omega (\theta ) &=& \frac {\int _{\theta }^{\infty }[1-\Psi\left(\frac{r-\mu}{\sigma}\right)]\,dr} {\int _{-\infty }^{\theta }\Psi \left(\frac{r-\mu}{\sigma} \right) \,dr} \\
& = & \frac {\int _{\frac{ \theta-\mu}{\sigma} }^{\infty }[1-\Psi\left(x\right)]\, \sigma \, dx } {\int _{-\infty }^{\frac{ \theta-\mu}{\sigma} }\Psi \left(x \right) \, \sigma  \,dx} \nonumber \\
& = & \frac {\int _{-S(\theta)}^{\infty }[1-\Psi\left(x\right)]\,dx } {\int _{-\infty }^{-S(\theta) }\Psi \left(x \right) \, dx} 
\end{eqnarray}
\noindent
where $S(\theta) = \frac{ \mu - \theta}{\sigma}$. Using an integration by parts, we have
\begin{eqnarray}
\int _{-\infty }^{a}\Psi \left(x \right) \, dx &=& \left[ x \Psi(x) \right]_{-\infty }^{a} - \int _{-\infty }^{a} x \psi \left(x \right) \, dx \nonumber \\
&=& a \Psi(a)  + \psi \left(a  \right) 
\end{eqnarray}
Using this result and the fact that $\psi(u)=\psi(-u)$, $1-\Psi\left(x\right)=\Psi\left(-x\right)$, we have:
\begin{eqnarray}
\int _{-\infty }^{-S(\theta)  }\Psi \left(x \right) \, dx &=& \psi(S(\theta) ) -S(\theta)  \Psi(-S(\theta) ) \\
\int _{-S(\theta) }^{\infty }[1-\Psi\left(x\right)]\,dx  &=& \psi(S(\theta) ) -S(\theta)  \Psi(-S(\theta) ) + S(\theta)  \nonumber
\end{eqnarray}
Hence the result
\begin{eqnarray}
\Omega (\theta ) = 1 + \frac{1}{ \frac{\psi(S(\theta))}{S(\theta)} - \Psi(-S(\theta))}
\end{eqnarray}
which concludes the proof \qed

\subsection{Proof of Proposition \ref{prop2}}\label{app:proof2}
As a function of $S(\theta)$, the $\Omega (\theta )$ ratio 's derivative with respect to $S(\theta)$ is given by
\begin{eqnarray}
\frac{\partial}{\partial S(\theta)} \Omega (\theta ) =\frac{ \psi(S(\theta)) }{ S(\theta)^2 \left( \frac{\psi(S(\theta))}{S(\theta)} - \Psi(-S(\theta)) \right)^2}
\end{eqnarray}
The strict positiveness of the derivatives concludes the proof that $\Omega (\theta ) $ ratio is an increasing function of $S(\theta)$. \qed

\subsection{Proof of Proposition \ref{prop3}}\label{app:proof3}
With notations of Proposition \ref{prop3}, we can remark that $\psi(u) = \psi(-u)$, $1-\Psi(u) = \Psi(-u)$, $\Psi^{'}(u) = \psi(u)$ and that $G(u)$ is defined (see equation \eqref{eq:Gdef}) such that 
\begin{eqnarray}
\frac{\partial}{\partial u}G( -\frac{1}{2} u^2) = - u \psi(u) 
\end{eqnarray}

We can also remark that the proper definition of the $\Omega (\theta )$ ratio implies (see proposition \ref{prop0}) that $\Psi(v)$ is dominated by some function
 $\frac{1}{|v|^{\alpha}}$ for some $\alpha>1$ for large negative value of $v$. In particular, we have that $\lim_{v \rightarrow -\infty} v \Psi(v) = 0$. The fact
 that $\Psi(v)$ is dominated by some function $\frac{1}{|v|^{\alpha}}$ for some  $\alpha>1$  for large negative value of $v$ implies also that $g$ is dominated
 by $\frac{1}{|v|^{\frac{\alpha+1}{2}}}$ for some $\alpha>1$ for large negative value of $v$, hence the function $G$ defined by $G(u)= \int_{-\infty}^u C g(x)\, dx$ 
is well defined. We also have that $\lim_{v \rightarrow -\infty} v \psi(v) = 0$.

We can therefore compute the following integral by integration by parts:
\begin{eqnarray}\label{eq:int_byparts}
\int _{-\infty }^{a}\Psi \left(x \right) \, dx &=& \left[ x \Psi(x) \right] _{-\infty }^{a} - \int _{-\infty }^{a} x \psi(x)  \, dx  \nonumber \\
&=& a \Psi(a) + G(-\frac{1}{2} a^2)
\end{eqnarray}
where we have used that $\lim_{v \rightarrow -\infty} v \Psi(v) = 0$ and $ \lim_{v \rightarrow -\infty} v \psi(v) = 0$.

Let us come back to our initial problem about computing the $\Omega (\theta )$ ratio. We have
\begin{eqnarray}
\Omega (\theta ) &=& \frac {\int _{\theta }^{\infty }[1-\Psi\left(\frac{r-\mu}{\sigma}\right)]\, \frac{dr}{\sigma}} {\int _{-\infty }^{\theta }\Psi \left(\frac{r-\mu}{\sigma} \right) \, \frac{dr}{\sigma} } \\
& = & \frac {\int _{\frac{ \theta-\mu}{\sigma} }^{\infty }[1-\Psi\left(x\right)]\, dx } {\int _{-\infty }^{\frac{ \theta-\mu}{\sigma} }\Psi \left(x \right) \, dx} \nonumber \\
& = & \frac {\int _{-\infty }^{S(\theta)} \Psi\left( x\right) \,dx } {\int _{-\infty }^{-S(\theta) }\Psi \left(x \right) \, dx}  
\end{eqnarray}
\noindent
where $S(\theta) = \frac{ \mu - \theta}{\sigma}$.  Using the intermediate result \eqref{eq:int_byparts}, we have 
\begin{eqnarray}
\Omega (\theta ) &= & \frac{ S(\theta) \Psi(S(\theta)) + G(-\frac{1}{2} S(\theta)^2)}{ -S(\theta) \Psi(-S(\theta)) +G(-\frac{1}{2}  S(\theta)^2)} \\
&= & 1 +\frac{1}{ \frac{G(-\frac{1}{2}  S(\theta)^2 )}{S(\theta)} -\Psi(-S(\theta))}
\end{eqnarray}
which concludes the proof \qed

\subsection{Proof of Proposition \ref{prop4}}\label{app:proof4}
Under the elliptical assumption for the distribution, proposition \eqref{prop3} applies. 
As a function of $S(\theta)$, the $\Omega (\theta )$ ratio 's derivative with respect to $a$ is given by
\begin{eqnarray}
\frac{\partial}{\partial S(\theta)} \Omega (\theta ) & =& \frac{\partial}{\partial S(\theta)} \frac{1}{ \frac{G(-\frac{1}{2}  S(\theta)^2 )}{S(\theta)} -\Psi(-S(\theta))} \\
& =& - \frac{\frac{ \frac{\partial}{\partial S(\theta)}G(-\frac{1}{2}  S(\theta)^2 )}{ S(\theta)} -   \frac{G(-\frac{1}{2}  S(\theta)^2 )}{S(\theta)^2}+\psi(-S(\theta))}{ \left( \frac{G(-\frac{1}{2}  S(\theta)^2 )}{S(\theta)} -\Psi(-S(\theta)) \right)^2 } \nonumber\\ 
& =& \frac{ G(-\frac{1}{2}  S(\theta)^2 )  }{ S(\theta)^2 \left( \frac{\psi(S(\theta))}{S(\theta)} - \Psi(-S(\theta)) \right)^2}
\end{eqnarray}
where we have used again that 
\begin{eqnarray}
\frac{\partial}{\partial u}G( -\frac{1}{2} u^2) = - u \psi(u) 
\end{eqnarray}
and the parity of the function $\psi$: $\psi(-u) = \psi(u)$. The strict positiveness of the derivatives concludes the proof that $\Omega (\theta ) $ ratio is an increasing function of $S(\theta)$. \qed
\end{document}